\documentclass[onecolumn, draftclsnofoot,12pt]{IEEEtran}
\IEEEoverridecommandlockouts
\usepackage{cite}
\usepackage{amsmath,amssymb,amsfonts}
\usepackage{algorithmic}
\usepackage{graphicx,color}
\usepackage{textcomp}
\usepackage{bm}
\newtheorem{theorem}{Theorem}
\newtheorem{lemma}{Lemma}
\newtheorem{example}{Example}

\newcommand{\F}{\mathbb{F}}

\newcommand{\Z}{\mathbb{Z}}
\newcommand{\sq}{\mathbf{s}}
\newcommand{\sa}{\mathbf{a}}

\newcommand{\g}{\hat{g}}

\newcommand{\im}{\textrm{Img}}
\newcommand{\ke}{\textrm{Ker}}

\def\BibTeX{{\rm B\kern-.05em{\sc i\kern-.025em b}\kern-.08em
    T\kern-.1667em\lower.7ex\hbox{E}\kern-.125emX}}
\begin{document}
\bibliographystyle{IEEEtran}

\title{Linear Complexity  of  A Family of Binary $pq^2$-periodic Sequences From Euler Quotients}

	\author{Jingwei Zhang, Shuhong Gao and Chang-An~Zhao$^*$
	\thanks{The work of Chang-An Zhao is partially supported by National Key R$\&$D Program of China under Grant No. 2017YFB0802500, by NSFC under Grant No. 61972428 and by the Major Program of Guangdong Basic and Applied Research under Grant No. 2019B030302008.  The work of Shuhong Gao was partially supported by the  National Science Foundation  under grants  DMS-1403062 and  DMS-1547399.
		The work of Jingwei Zhang was partially supported by the National Social Science Fund of China under Grant No.14BXW031 and by Guangdong Basic and Applied Basic Research Foundation under Grant No. 2019A1515011797.
		
		J.W. Zhang is with School of Information Science, Guangdong University of Finance and Economics, Guangzhou, 510320, P.R. China (E-mail: jingweizhang@gdufe.edu.cn)  }%
	\thanks{S. Gao is with School of Mathematical and Statistical Sciences, Clemson University, 
Clemson, S.C.,29634, U.S.A (E-mail: sgao@math.clemson.edu)}
\thanks{C.-A, Zhao is with  School
	of Mathematics, Sun Yat-sen University, Guangzhou 510275, P.R.China and with 	Guangdong Key Laboratory of Information Security,	Guangzhou {\rm 510006}, P.R. China.
	
	$*$ Corresponding author.
	(E-mail: zhaochan3@mail.sysu.edu.cn)} 

}
\maketitle

\begin{abstract}
	We first introduce a 
 family of binary $pq^2$-periodic sequences based on  the  Euler quotients modulo $pq$, where $p$ and $q$ are  two distinct odd primes and $p$ divides $q-1$. The minimal polynomials and linear complexities are determined for the proposed sequences provided that $2^{q-1} \not\equiv 1 \mod{q^2}.$ The results show that the proposed sequences have high linear complexities. 
\end{abstract}


\section{\textsf{Introduction}}\label{sec1}
We will begin by the following definition of the Euler quotient modulo a product of two distinct odd primes.  
Let $p$ and $q$ be two distinct odd  primes.
For a nonnegative integer $t$ that is relatively prime to $pq$, the Euler quotient $\psi(t) \pmod{pq}$ is defined as a unique integer in $\Z_{pq}$ with
\begin{equation}\label{equa1}
\psi(t)=\frac{t^{\varphi(pq)}-1}{pq} \pmod{pq},
\end{equation}
where $\varphi(\cdot)$ is the well-known Euler-phi function. We also define $\psi(t)=0$ if  $t$ and $pq$ are not relatively prime. 

It can be seen easily that  the Euler quotient has the following property: 
%
\begin{equation}\label{eqn3}
\psi(t+kpq)\equiv \psi(t) {\color{black} + } kt^{-1}(p-1)(q-1) \pmod{pq}.
\end{equation}
where $t,k\in{\Z}$ and  $t$ is relatively prime to $pq.$

In 2010,  Chen, Ostafe and  Winterhof\cite{ChenOstafeWin2010} introduced families of binary sequences using Fermat/Euler quotients. Then several nice cryptographic properties of these sequences were proved in \cite{ ChenWinterhof2012a,ChenWinterhof2012b, Chen2013,Chen2014,Chen2015a,Chen2015b}.
Based on the distribution and algebraic structure of the Fermat quotients, the linear complexity was determined for a binary threshold sequence defined from Fermat quotients~\cite{Chen2013}. Naturally, the definition of the Euler quotient can be generalized by  Euler's Theorem~\cite{AGOH1997}. Chen and Winterhof extended the distribution of pseudorandom numbers and vectors derived from Fermat quotients to Euler quotients~\cite{ChenWinterhof2012a}. Moreover, linear complexities were calculated for  binary sequences derived from Euler quotients with prime-power modulus. 
%
Trace representations and linear complexities were investigated for binary sequences derived from Fermat quotients~\cite{Chen2014}. Subsequently, a trace representation was given for a family of  binary sequences derived from Euler quotients modulo a fixed power of a  prime~\cite{Chen2015b}. Chen and Winterhof generalized Fermat quotients to  the so-called polynomial quotients in~\cite{ChenWinterhof2012b}. Then the $k$-error linear complexity was determined for binary sequences derived from the polynomial quotient modulo a prime~\cite{Chen2015a} or its power~\cite{Niu2016}, respectively. In~\cite{Song2016}, a series of optimal families of perfect polyphase sequences were derived from the array structure of Fermat-quotient sequences. 
All of the above results show that pseudorandom sequences derived from Fermat quotients, Euler quotients or their variants can be regarded as an important class of sequences  from a cryptographic point of view.

In this paper, we study binary sequences derived from the Euler quotient modulo $pq$. 
Using the same notation as above,  
a binary threshold sequence $\sq =\{s_t|t\in \Z, t\geq 0 \}$ from the Euler quotient modulo $pq$ can be defined as
\begin{equation}\label{eqn1}
s_t= \left\{\begin{array}
{l@{\quad\quad}l}
0,  &\textrm{if} ~~~ 0\leq\frac{\psi(t)}{pq}<\frac{1}{2},\\
1, & \textrm{if} ~~~ \frac{1}{2}\leq\frac{\psi(t)}{pq}<1.
\end{array}\right.
\end{equation}

For our purpose, we introduce the concept of the linear complexity of binary sequences now. The linear complexity  of an $N$-periodic sequence $\sa=\{a_i|i\in \Z, i\geq 0\}$ over the binary field $\F_2$ is the smallest nonnegative integer $L$ for which there exist elements $c_1,c_2,\cdots, c_L\in \F_2$ such that 
$$a_i+c_1\cdot a_{i-1} + \cdots + c_L\cdot a_{i-L}=0, ~for~all~ i\geq L.$$

Let $A(x)=\sum\limits_{i=0}^{N-1}a_{i}x^{i}\in \F_2[x]$ be the generating polynomial of $\sa$.  By~\cite{Ding1991}, the minimal polynomial of $\sa$   is defined as
$$M_{\sa}(x)=  \frac{x^N-1}{\gcd(x^N-1,A(x))},$$
where $\gcd(\cdot,\cdot)$ denotes the greatest common divisor of two polynomials over $\F_2$ and the linear complexity of $\sa$ is  
$$\mathcal{L}(\sa)=N - \deg(\gcd(x^N-1,A(x))). 
$$
Note that the linear complexity is of fundamental importance as a complexity measure for binary sequences in sequences designs~\cite{Ding1991,fan1996sequence,Golomb2005}.
Besides the measure of the linear complexity for sequences, other measures are also required according to different specific requirements from  applications, for example,  low autocorrelations or cross-correlations~\cite{WeiSu2018,YangIT2018}, good nonlinear properties~\cite{Zhao2018,JiangLin2018,Zhang2019}, and sphere complexities (or $k$-error complexities)~\cite{chen2015,chen2018}.
{\color{black}
For a binary sequence to be cryptographically strong, the linear complexity of the sequence should be at least a half of the least period of the sequence. 
In fact, if the linear complexity of a binary sequence $\sa$ for additive stream ciphering purposes is $\mathcal{L}(\sa)$, then  
$2\mathcal{L}(\sa)$ consecutive bits of the sequence can be employed to ``recover" the whole key stream with the well-known Berlekamp-Massey algorithm~\cite{Massey1969,Ding1991}. Therefore, 
it is  necessary that  key stream sequences must have large linear complexity in additive stream ciphers  to resist the Berlekamp-Massey attack from a  point of cryptographic view.
}

The main contribution of this paper is to determine the minimal polynomial and the linear complexity of the sequence defined in (\ref{eqn1}). We state our main result as follows.

\begin{theorem} \label{th1}
	Let $p$ and $q$ be two distinct odd primes with $p$ dividing $q-1$.  Assume that  $2^{q-1}\not \equiv 1 \pmod{q^2}.$ Then the  binary threshold sequence $\sq$  defined in~(\ref{eqn1}) has period at least $pq^2.$
	The minimal polynomial of $\sq$ is 
	\[M_{\sq}(x)= \left\{\begin{array}
	{l@{\quad\quad}l}
	\Phi_{pq^2}(x), 
	&\textrm{if} ~~~ q \equiv 1 \pmod{4}, \\
	\Phi_{pq^2}(x)\Phi_{pq}(x),  &\textrm{if} ~~~ q \equiv 3 \pmod{4},
	\end{array}\right.\] 		
	where $\Phi_{n}(x)$ denotes the $n$-th cyclotomic polynomial  for any positive integer $n$ and 
	the linear complexity of $\sq$ 
	is 
	\[\mathcal{L}(\mathbf{s})= \left\{\begin{array}
	{l@{\quad\quad}l}
	(p-1)(q^2-q), &\textrm{if} ~~~ q \equiv 1 \pmod{4}, \\
	(p-1)(q^2-1),  &\textrm{if} ~~~ q \equiv 3 \pmod{4}.
	\end{array}\right.\]	
\end{theorem}
To the best of our knowledge, this is the first time to introduce this kind of sequences based on the Euler quotient modulo a product of two distinct odd primes. 
Under the condition that $p$ divides $q-1$, we will show that  the binary sequence has period at least $pq^2$. Furthermore, minimal polynomials and linear complexities of this class of binary sequences are determined. {\color{black}
It turns out that the proposed sequence has  linear complexity which is much higher than a half of the least period. This means that it can be employed in stream ciphers. Note that the sphere complexity~\cite{Ding1991} (or $k$-error linear complexity~\cite{StampMartin1993}) of key stream sequences which measures the stability of linear complexity is also cryptographically important. Several classes of sequences derived from Euler (or Fermat) quotients modulo a fixed prime power have good $k$-error linear complexities (see~\cite{Chen2015a,Niu2016}). We expect that the proposed sequence also posses such good properties. In fact, some experimental results have indicated that $1$-error linear complexity of the presented sequence here is still equal to the original linear complexity, which means that it is entirely possible to have 
good linear complexity stability for the proposed sequence.  We leave this problem as future work. }

By using the generalized cyclotomic techniques, one can also construct other binary sequences with period $pq^2$. We refer the reader to see \cite{DING1998140,Bai2005,HuYueWang2012}
for more details. We emphasize that our results are new.  In particular, we  point out that our results are not  special cases of Theorem 4.2 of \cite{HuYueWang2012} although both may give a sequence with period $pq^2$.
In fact, this can be seen easily by comparing linear complexities of the two families of binary sequences. 

In the rest of the paper, we give a proof of the above theorem  in Section~\ref{sec2}, and conclude with a few remarks in Section~\ref{sec4}.

\section{Proof of Main Results} \label{sec2}

In this section, we are devoting to the proof of the main results. 

We first show that $pq^2$ is one of the periods of  sequence $\mathbf{s}$ under the condition that $p$ is a divisor of $q-1$. 
Setting $k=q$ in (\ref{eqn3}), we see that 
$$\psi(t+pq^2)=\psi(t) \pmod{pq}$$
which implies $s_{t+pq^2}=s_t$ for all $t\geq 0$.
Thus the sequence $\sq$ is   periodic with period $pq^2$. We will demonstrate that $pq^2$ is the least  period of the sequence $\sq$ in the following lemma.

\begin{lemma}\label{lma1}
	With the notation above, the sequence $\sq$ has period at least $pq^2$.
\end{lemma}
\begin{proof}
	We first prove that $pq$ is not a period of the  sequence $\sq$.	
	By (\ref{eqn3}), we have
	\begin{equation*}
	\psi(pq+1)\equiv \psi(1)+(p-1)(q-1)\equiv (p-1)(q-1) \pmod{pq}.
	\end{equation*}
	It follows from $2(p-1)(q-1)-pq=(p-2)(q-1)-p>0$ that the $(pq+1)$-th term of the sequence $\sq$ is equal to 1, i.e.,  $s_{pq+1}=1$. Note that $s_1=0\neq 1=s_{pq+1}$ according to the definition of the sequence $\sq$. Hence $pq$ is not a period of the sequence $\sq$.
	
	Now we prove that $q^2$ is not a period of  the sequence $\sq$. We can assume that $q^2$ is a period of  the sequence $\sq$.  Let $k=\lceil \frac{pq}{2(p+q-1)}\rceil$. It follows from (\ref{eqn3}) that
	$$\psi(-1 +kpq)=k(p+q-1) \pmod{pq}$$ and thus	
	$s_{-1+kpq}=1$.  This means that the sequence $\sq$ satisfies  $s_{-1+kpq+q^2}=s_{-1+kpq}=1$.  However,  
	we have $s_{-1+kpq+q^2}=0$ according to the definition of the sequence $\sq$ and $\gcd(kpq-1+q^2, pq)=\gcd(kpq+(q-1)(q+1), pq)= p.$ It follows that $s_{kpq-1+q^2}\neq s_{kpq-1}$, a contradiction. 
	
	Hence the least period of the sequence $\sq$ is $pq^2,$ which completes the proof of the lemma. 
\end{proof}

For any  integer $n\geq 2$, we denote by  $\Z_n=\{0,1,\cdots, n-1\}$ all representatives for the residue classes of integers modulo $n$ and by $\Z_n^*$ all representatives that are relatively prime to $n$ in $\Z_n,$ respectively. 
Since the least period of $\sq$ is $pq^2$,  we  restrict the action of $\psi$ on $\Z_{pq^2}$ sometimes. With a slight abuse of notation, we shall still use the same symbol $\psi$ to denote this restriction of the Euler quotient on  $\Z_{pq^2}$.

Let $g\in \Z_{pq^2}^*$ be a fixed common primitive root of both $p$ and $q^2$. The Chinese Reminder Theorem(CRT) \cite{dingyi1996chinese} guarantees that there exists  an element  $h$ of $\Z_{pq^2}^*$ such that
\[\left\{\begin{array}
{l@{\quad\quad}l}
h\equiv g \pmod{p},  \\
h\equiv 1 \pmod{q^2}.
\end{array}\right.\]

Put  $d=\gcd(p-1,q-1)$ and
$e=\textrm{lcm}(p-1,q-1)=(p-1)(q-1)/d$ where $\textrm{lcm}$ denotes least common multiple.
Then the unit group $\Z_{pq^2}^*$ of the ring $\Z_{pq^2}$ \cite{DING1998140} can be written as follows 
\begin{equation}
\\Z_{pq^2}^{*}=\{g^{i}h^{j}: 0 \leq i < qe, 0\leq j < d\}.
\end{equation}

The following lemma shows that the map $\psi$ is a group homomorphism when we restrict the action of the map $\psi$ to the unit group   $\Z_{pq^2}^*$. 
\begin{lemma}\label{lma2}
	\textit{	Let $\psi: t\rightarrow \psi(t)$ be the map from $\langle\Z_{pq^2}^*,\cdot\rangle$ to $\langle p\Z_{pq},+\rangle,$ where $p\Z_{pq}=\{lp\mid 0\leq l\leq q-1\}$ contains exactly all of the residue classes which are divisible by $p$ in the addition group $\Z_{pq}$.
		Then $\psi$ is a surjective group homomorphism.  }
	
	\textit{	Let $g$ and $h$ be defined as above. Then the image and kernel of $\psi$ are given as} $$\im(\psi)=p\Z_{pq}$$ \textit{and}  $$\ke(\psi)=\langle g^{q},h \rangle =\{g^{qi}h^j\mid 0\leq i<e, 0\leq j<d\},$$
	\textit{respectively.}
\end{lemma}
\begin{proof}
	Note that $t^{p-1}\equiv 1 \pmod p$ for $t\in \Z_{pq^2}^*$. We can write $t^{p-1}=1+t{'}p$ for some integer $t'$. Substituting it into~(\ref{equa1}), we have
	$$\psi(t)=\frac{(1+t{'}p)^{q-1}-1}{pq}\equiv t{'}(q-1)q^{-1}\equiv 0 \pmod p$$
	as $p$ divides $q-1$. This means that $\psi(t)$ is divisible by $p$ and thus the map $\psi$ is well defined.
	
	For $u,v\in \Z_{pq^2}^*$ it follows from 
	 Euler's  Theorem that
	\begin{equation*}
	\begin{aligned}
	\psi(uv)=& \frac{(uv)^{\varphi{(pq)}}-1}{pq} \\
	=&\frac{(uv)^{\varphi{(pq)}}-u^{\varphi{(pq)}}+u^{\varphi{(pq)}}-1}{pq}\\
	=&u^{\varphi(pq)}\psi(v)+\psi(u)\\
	\equiv& \psi(u)+\psi(v)\pmod{pq}
	\end{aligned}
	\end{equation*}	
	which yields the map $\psi$ is a group homomorphism.
	
	Now we show that  the map $\psi$ is surjective. There exists some integer $t_1$ such that $g^{q-1}=1+t_1q$ with $t_1\not \equiv 0 \pmod q$ since $g$ is a primitive root in $\Z_{q^2}^*$. This implies that
	\begin{equation*}
	\begin{aligned}
	\psi(g)=&\frac{g^{\varphi(pq)}-1}{pq}\equiv \frac{(1+t_1q)^{p-1}-1}{pq}\\
	\equiv&
	t_1(p-1)p^{-1}\not\equiv 0 \pmod q.
	\end{aligned}
		\end{equation*}
	
	Note that $\psi(g)\equiv 0 \pmod p$. It follows from the CRT that  there exists some positive integer $a$ with $a\in \Z_q^*$ such that $$\psi(g)\equiv pa\not \equiv 0 \pmod{ pq}.$$
	It follows that $pa$ is one generator of the addition group $p\Z_{pq}$. Consequently, the map $\psi$ is surjective and  $\im(\psi)=p\Z_{pq}$.
	
	It is known that both $\psi(g)$ and   $\psi(h)$ are divisible by $p$. Also,
	$$\psi(g^q)=q\psi(g) = 0 \pmod q.$$
	On the basis of the CRT, we have $\psi(g^q)=0 \pmod{pq}$.
	Hence $g^q \in \ke(\psi)$. Observe that $h\equiv 1 \pmod{q^2}$. We can write $h=1+q^2h_1$. Hence
	$$\psi(h)= \frac{h^{\varphi(pq)}-1}{pq} \equiv p^{-1}\frac{(1+q^2h_1)^{\varphi(pq)}-1}{q} \equiv 0\!\! \pmod q\!\!.$$
	Combining the above equation with $\psi(h)=0\pmod p$, we get   $h\in \ke(\psi).$  Therefore, we have
	$$\{(g^q)^ih^j\!\!\!\!\!\pmod{pq^2}\!\!\mid\!\! 0\leq i<e,\! 0\!\leq\! j<d\}\!=\!\langle \!g^{q},\!h \!\rangle\! \subseteq\! \ke(\psi).$$
	
	Now we need to show that the kernal $\ke(\psi)$ 
	and the subgroup $\langle g^{q},h \rangle$ have the same cardinality. By the Third Isomorphism Theorem~\cite{Artin2011}, we have 
	$$\Z_{pq^2}^*/\langle g^{q},h \rangle \simeq   \bm{(}\langle g,h \rangle/\langle h \rangle\bm{)} / 
	\bm{(}\langle g^{q},h \rangle/\langle h \rangle \bm{)} \simeq  \langle g \rangle /\langle g^{q} \rangle. $$
	This yields that $g^0\langle g^{q},h \rangle, g^1\langle g^{q},h \rangle,\cdots, g^{q-1}\langle g^{q},h \rangle$ are all cosets of the subgroup  $\langle g^{q},h \rangle$ of $\Z_{pq^2}^*$. It follows  that 
	$|\langle g^{q}, h  \rangle|=(p-1)(q-1)$. On the other hand, 
	according to the Fundamental Homomorphism Theorem~\cite{Artin2011}, we see that
	$$|\ke(\psi)|=|\frac{\Z_{pq^2}^*}{\im(\psi) }|=\frac{(p-1)(q-1)q}{q}=(p-1)(q-1)$$
	and so $\langle  g^{q}, h \rangle = \ke(\psi)$. This completes the proof of the lemma. $ $ \end{proof}	
Note that Lemma~\ref{lma2} gives that $\psi(g)=pa \pmod{pq}$ with some $a\in \Z_q^*$. This means that $\psi(g)=pa \pmod q$ by the CRT.
Let $b$ be the inverse of $a$ in $\Z_q^*$, i.e., $ab\equiv 1 \pmod q$. Define $\g=g^b$ in $\Z_{pq^2}^{*}$. Then
$$\psi(\g)=b
\cdot \psi(g) \pmod{pq}$$ by the homomorphism property of the map  $\psi$. It follows from $\psi(g)\equiv pa \pmod q$  that $$\psi(\g)\equiv pab\equiv p \pmod q.$$ Combining the above equality with $\psi(\g)=0 \pmod p$, we get $\psi(\g)=p \pmod{pq}$. 
The following lemma describes a partition of $\Z_{pq^2}^*$ which will give a new explanation of the definition of the sequence $\sq$. 

\begin{lemma}\label{lma3}
	Let $\g$ be an element in $\Z_{pq^2}^{*}$ with $\psi(\g)=p \pmod{pq}.$
	Define
	\begin{equation*}
	D_{\ell}=\{t: \psi(t)=p\ell \pmod{pq}, t\in{\\Z_{pq^2}^{*}}\}
	\end{equation*} and
	$$\hat{D}_{\ell}=\g^\ell D_0=\{\g^\ell\cdot t \pmod{pq^{2}}: t\in D_0\}$$ for $\ell=0, 1,\cdots, q-1.$
	Then $\Z_{pq^2}^{*}=\bigcup\limits_{\ell=0}^{q-1}D_{\ell}$ and $D_{\ell}=\hat{D}_{\ell}
	$ for all $\ell \in \Z_q$.

\end{lemma}
\begin{proof}	
	We first prove that $\hat{D}_{\ell}= D_{\ell}$ for all $\ell\in \Z_q$. Note that Lemma~\ref{lma2} gives that $D_{0}=\ke(\psi).$ It is easy  to see that for $\g^\ell t_0\in \hat{D}_{\ell}$ with $t_0\in D_0=\ke(\psi)$. We have
	$$\psi(\g^\ell t_0)=l\cdot \psi(\g)+\psi(t_0)=\ell p \pmod{pq}. $$	This implies that $\hat{D}_{\ell}\subseteq D_{\ell}$. Conversely, for $t\in D_{\ell}$, we have 
	$$\psi(t)=pl=l\psi(\g)=\psi(\g^\ell) \pmod{pq}$$
	and thus $$\psi(\frac{t}{\g^\ell})=0 \pmod{pq}$$ by
	the homomorphism property  of $\psi$. This means that $$\frac{t}{\g^\ell}\in \ke(\psi)=D_0.$$  Therefore, there exists some element $t_0\in D_0$ such that
	$$\frac{t}{\g^\ell }\equiv t_0 \pmod{pq^{2}}.$$ Hence we have $t=\g^\ell\cdot t_0\in \g^\ell D_0=\hat{D}_{\ell}$ and so ${D}_{\ell}=\hat{D}_{\ell}$.  This completes the whole proof of the lemma.$ $
\end{proof}

%
%


By the definition of $D_{\ell}$ and $\hat{D}_{\ell}$ , Lemma~\ref{lma3} gives that $|D_{\ell}|=|\hat{D}_{\ell}|=(p-1)(q-1)$ for $\ell=0, 1,\cdots, q-1.$
Let $P= \{t: t\in Z_{pq^2}, \gcd(t,pq)\neq 1\}$. The sequence $\sq$ can be rewritten as 
\[s_t= \left\{\begin{array}
{l@{\quad\quad}l}
0  &\textrm{if} ~~~ t\in D_0\cdots \cup D_{(q-1)/2}\cup P,\\
1 & \textrm{if} ~~~ t\in  D_{(q+1)/2}\cup \cdots D_{q-1}.
\end{array}\right.\]
The new explanation of the sequence $\sq$ will be helpful to determine linear complexities. 
We will make extensive use of the following  lemmas for completing the proof of Theorem~\ref{th1}.

\begin{lemma}\label{lma4}
	\textit{	
		For any $0\leq i < q,$ if $u\pmod{pq^2}\in{D_{j}}$ for some $0\leq j< q,$ we have
		\begin{equation*}
		uD_{i}=\{uv \pmod{pq^2}: v\in{D_{i}}\}=D_{i+j}.
		\end{equation*}
		where all the subscripts are certainly understood  modulo $q$. In particular, $D_{ri}=\g^{r}D_{i}$ for $0\leq r \leq q-1$.
	}
\end{lemma}
\begin{proof} If $v\in D_i$, then $u=\g^ju_0$ and  $v=\g^iv_0$ with $u_0,~v_0\in D_0$. Hence
	$uv=\g^{i+j}u_0v_0\in \g^{i+j}D_0=D_{i+j}.$
	This implies that $uD_i\subseteq D_{i+j}$. Conversely, it can be seen easily that $D_{i+j}\subseteq uD_i$. This finishes the proof of the lemma.$ $
\end{proof}

The study of the behavior of the coset $D_{\ell}$  modulo various divisors of $pq^2$ leads to a number of useful lemmas. 

\begin{lemma}\label{lma5}
	\textit{		
		For $0\leq \ell < q,$ we have the following two multiset equalities
		$$\{u \pmod{p}: u\in{D_{\ell}}\}= (q-1)*\Z_{p}^*,$$ where   $(q-1)*\Z_{p}^*$ is the multiset in which each element of $\Z_{p}^*$ appears with multiplicity
		$q-1$, and $$\{u \pmod{q}: u\in{D_{\ell}}\}=(p-1)*\Z_{q}^*,$$ where   $(p-1)*\Z_{q}^*$ is the multiset in which each element of $\Z_{q}^*$ appears with multiplicity
		$p-1$.}
\end{lemma}
\begin{proof}Note that 
	 $u\in{D_{\ell}}$ can be written as $u=\g^{\ell}g^{qi}h^{j}$ for $0\leq i < e$ and $0\leq j < d.$ Recall that $\g=g^b$ with some fixed $b\in \Z_q^*$ in Lemma~\ref{lma3}. Then $u=g^{qi+b\ell}h^j$ in $\Z_{pq^2}$ and so
	%
	$$u=\g^{\ell}g^{qi}h^{j}\equiv g^{qi+b\ell+j}\equiv g^{b\ell+j}\cdot (g^q)^i \pmod{p}.$$
	According to $\gcd{(p-1, q)}=1,$ we see that $g^q$ is also a primitive root of $\Z_p^*$.
	If we fix some $j_0\in\Z_d$, then $u\equiv g^{qi+b\ell +j_0}\pmod{p}$ runs through $\Z_p^*$ when $i$ runs through $\Z_e$. Now we count the multiplicity of each element in $\Z_p^*$ when
	$i$ and $j$ run through $\Z_e$ and $\Z_d$   respectively. Assume that
	$$u\equiv g^{qi+b\ell +j_0}\equiv g^{a_0}\pmod{p}$$
	where $0\leq a_0 \leq p-2.$
	This means that $$qi=a_0-j_0-b\ell \pmod{p-1}$$ for $i\in \Z_e$.
	According to $\gcd{(p-1, q)}=1,$ it is equivalent to
	$$i\equiv q^{-1}(a_0-j_0-b\ell) \pmod{p-1}.$$
	There exist $\frac{q-1}{d}$ many solutions in the form of
	$i_0, i_0+(p-1), \cdots, i_0+(\frac{q-1}{d}-1)(p-1).$
	Note that $j_0$ has $d$ choices. This implies that there are  $(q-1)$ many elements of $D_{\ell}$ mapping into one element in $\Z_{p}^*.$
	In a similar manner, we can prove the second multiset equality in the lemma. This completes the  proof of the lemma. $ $
	%
\end{proof}

\begin{lemma}\label{lma6}
	\textit{For $0\leq \ell < q,$ we have
		$$\{u \pmod{pq}: u\in{D_{\ell}}\}=\Z_{pq}^*.$$
	}
\end{lemma}
\begin{proof} It is obvious that the map from $D_{\ell}$ to $\Z_{pq}^*$ with $u\rightarrow u\pmod{pq}$ is well-defined. Thus
	it is sufficient to prove that the map is  one-to-one since both $D_{\ell}$ and $\Z_{pq}^*$ have the same cardinality.
	
	For $u_1,~u_2\in{D_{\ell}},$ we write $u_1=\g^{\ell}g^{qi_1}h^{j_1}$ and $u_2=\g^{\ell}g^{qi_2}h^{j_2}$ with $i_1,~i_2\in \Z_{e}$ and $j_1,~j_2\in \Z_d$ respectively. Assume that
	$$\g^{\ell}g^{qi_1}h^{j_1}=u_1 \equiv u_2=\g^{\ell}g^{qi_2}h^{j_2}\pmod{pq}. $$
	We will illustrate that $i_1=i_2$ and $j_1=j_2$.
	
	Note that
	$$g^{qi_1}h^{j_1}=g^{qi_2}h^{j_2}\pmod{ pq} $$
	as $\gcd(\g,pq)=1$. It follows from the CRT that
	\[\left\{\begin{array}
	{l@{\quad\quad}l}
	g^{qi_1+j_1}=g^{qi_2+j_2} \pmod{p},  \\
	g^{qi_1}=g^{qi_2} \pmod{q}.
	\end{array}\right.\]
	This implies that
	\[\left\{\begin{array}
	{l@{\quad\quad}l}
	qi_1+j_1\equiv qi_2+j_2 \pmod{p-1},  \\
	qi_1 \equiv qi_2 \pmod{q-1}.
	\end{array}\right.\]
	Note that $d=\gcd(p-1,q-1)$.  It follows from the above equality that
	\[\left\{\begin{array}
	{l@{\quad\quad}l}
	qi_1+j_1\equiv qi_2+j_2 \pmod{d},  \\
	qi_1 \equiv qi_2 \pmod{d}.
	\end{array}\right.\]
	This gives that $$j_1\equiv j_2 \pmod d.$$
	Since $j_1$ and $j_2$ belong to $\Z_d$, we have $j_1=j_2$. In the following, we will show that $i_1=i_2$ on the basis of the fact that $j_1=j_2$. Now we have
	\[\left\{\begin{array}
	{l@{\quad\quad}l}
	qi_1\equiv qi_2 \pmod{p-1},  \\
	qi_1 \equiv qi_2 \pmod{q-1}.
	\end{array}\right.\]
	Since $\gcd(q,p-1)=\gcd(q,q-1)=1$, it follows that
	\[\left\{\begin{array}
	{l@{\quad\quad}l}
	i_1\equiv i_2 \pmod{p-1},  \\
	i_1 \equiv i_2 \pmod{q-1}.
	\end{array}\right.\]
	Recall $e=(p-1)(q-1)/d=\textrm{lcm}(p-1,q-1)$.
	It follows from the above equations that 
	$$i_1=i_2 \pmod e.$$
	Since $i_1$ and $i_2$ belong to $\Z_e$, we have $i_1=i_2$. This completes the proof of the lemma.
	$ $ 
\end{proof}

\begin{lemma}\label{lma7}
	Let $\g,~g\in \Z_{pq^2}^*$ be the same notations as above.
	For $0\leq \ell < q,$ we have the following multiset equality
	$$\{u \pmod{q^2}: u\in{D_{\ell}}\}=(p-1) * \mathfrak{\g}^{\ell}\langle \mathfrak{g}^{q}\rangle,
	$$ where $\mathfrak{\g}$ and $\mathfrak{g}$ denote $(\g\pmod{q^2})$ and $(g\pmod{q^2}),$ respectively. The set $\mathfrak{\g}^{\ell}\langle \mathfrak{g}^{q}\rangle$ is contained in $ \Z_{q^2}^*$ and  $(p-1) * \mathfrak{\g}^{\ell}\langle \mathfrak{g}^{q}\rangle$ is the multiset in which each element of $\mathfrak{\g}^{\ell}\langle \mathfrak{g}^{q}\rangle$ appears with multiplicity $p-1$.
\end{lemma}
\begin{proof}
	Note that
	$$u=\g^{\ell}g^{qi}h^{j}\equiv \mathfrak{\g}^{\ell}\mathfrak{g}^{qi}\cdot 1\equiv \mathfrak{\g}^{\ell}(\mathfrak{g}^q)^{i} \pmod{q^2}.$$
	This means that $(u \pmod{q^2})$ belongs to $\mathfrak{\g}^{\ell}\langle \mathfrak{g}^{q}\rangle$ indeed.  So the map from $D_{\ell}$ to $\mathfrak{\g}^{\ell}\langle \mathfrak{g}^{q}\rangle$ with $u\rightarrow u \pmod {q^2}$  is well-defined. Now we count the multiplicity when $u$ runs through the set $D_{\ell}$. Assume that
	$$\mathfrak{\g}^{\ell}(\mathfrak{g}^q)^{i}\equiv \mathfrak{\g}^{\ell}(\mathfrak{g}^q)^{a_0} \pmod{q^2}$$
	for some fixed $a_0\in \Z_{q-1}.$ It follows that
	$$qi\equiv qa_0 \pmod {q-1},$$
	i.e.,
	$$i\equiv a_0 \pmod{q-1}.$$
	There exist $\frac{p-1}{d}$ many solutions for $i\in \Z_{e}$ in the form of
	$a_0, a_0+(q-1), \cdots, a_0+(\frac{p-1}{d}-1)(q-1).$
	Note that $j\in \Z_d$ has $d$ choices. Altogether, there are $(p-1)$ many
	elements of $D_{\ell}$ mapping into one element in $\mathfrak{\g}^{\ell}\langle \mathfrak{g}^{q}\rangle.$ This finishes the proof of the lemma.
\end{proof}

Define  $D_{\ell}(x)=\sum\limits_{u\in{D_{\ell}}}x^{u}\in{\mathbb{F}_2[x]}.$
There exists an important connection between the polynomial $D_{\ell}(x)$ and the cyclotomic polynomial $\Phi_{n}(x)$ that will allow us to determine the minimal polynomial of sequences $\sq.$
\begin{lemma}\label{lma8}
	Let $\gamma$ be a fixed $pq^2$-th primitive root of unity in the algebraic closure of ${\F}_2$ and $v$  an element in $\Z_{pq^2}$.  Then
	\[
	D_{\ell}(\gamma^v)= 
	\left\{\begin{array}
	{l@{\quad\quad}l}
	1 &\textrm{if} ~~~ \gcd(v,pq^2)=q, \\
	0  &\textrm{if} ~~~ \gcd(v,pq^2)\in 
	\{p,pq,q^2\}
	\end{array}\right.\]
	and
	\[
	D_{\ell}(x) \equiv \left\{
	\begin{array}{l@{\quad\quad}l}
	1 &\pmod {\Phi_{pq}(x)},\\
	0 &\pmod {(\Phi_{p}(x)\Phi_q(x)\Phi_{q^2}(x))}.
	\end{array}
	\right.
	\]

	%
\end{lemma}

\begin{proof}
	We distinguish two cases according to the distinct value of the greatest common divisor of $v$ and $pq^2$.
	\begin{enumerate}
		\item 
		For $v\in \Z_{pq^2}$ with $\gcd(v,pq^2)=q$, it follows that $\gamma^v$ is a $pq$-th primitive root of unity . On the basis of Lemma~\ref{lma6}, we have
		$$D_{\ell}(\gamma^v)= \sum_{u\in D_{\ell}}\gamma^{uv}=\sum_{u\in \Z_{pq}^*}(\gamma^{v})^u.$$
		Note that $\sum_{u\in \Z_{pq}^*}\gamma^{uv}$ is equal to the sum of all $pq$-th primitive roots of unity that is also the coefficient of the second highest term of the cyclotomic polynomial $\Phi_{pq}(x)$. According to Exercise 2.57 of~\cite{Lidl1986}, we see that
		\begin{equation*}
		\begin{aligned}
		\Phi_{pq}(x)=&\frac{\Phi_{q}(x^p)}{\Phi_{q}(x)}=\frac{x^{p(q-1)}+x^{p(q-2)}+\cdots+1}{x^{q-1}+x^{q-2}+\cdots+1}\\
		=&x^{(p-1)(q-1)}+1\cdot x^{(p-1)(q-1)-1}+\cdots.
	\end{aligned}
		\end{equation*}
		This indicates that $$D_{\ell}(\gamma^v)= \sum_{u\in{D_{\ell}}}\gamma^{uv}=\sum_{u\in \Z_{pq}^*}\gamma^{uv}=1$$ for $v\in \Z_{pq^2}$ with $\gcd(v,pq^2)=q$.

		\item For $v\in \Z_{pq^2}$ with $\gcd(v,pq^2)=q^2$, it follows that $\gamma^v$ is a $p$-th primitive root of unity. It follows from Lemma~\ref{lma5} and the even parity of $(q-1)$ that
		$$\sum_{u\in D_{\ell}}\gamma^{uv}=(q-1)\sum_{u\in \Z_p^*}(\gamma^{v})^{u}=\equiv0 \pmod 2.$$
		
		For $v\in \Z_{pq^2}$ with $\gcd(v,pq^2)=pq \textrm{~or~} p$, then $\gamma^v$ is a $q$-th or $q^2$-th primitive root of unity respectively. Using a similar argument, it follows from Lemmas~\ref{lma5} and~\ref{lma7} and the even parity of $(p-1)$ that  $\sum_{u\in{D_{\ell}}}\gamma^{uv} = 0$ in this case. 
	\end{enumerate}
	It follows from the definition of cyclotomic polynomials that 
	\[
	D_{\ell}(x) \equiv \left\{
	\begin{array}{l@{\quad\quad}l}
	1 &\pmod {\Phi_{pq}(x)},\\
	0 &\pmod {\Phi_{n}(x)} ~\textrm{if}~ n= p,~q~  \textrm{or}~q^2.
	\end{array}
	\right.
	\]
	Therefore, we get the desired result since the cyclotomic polynomials $\Phi_{p}(x),\Phi_{q}(x)$ and $\Phi_{q^2}(x)$ over $\F_2$ are relatively prime.$ $
\end{proof}

We are now in a position to give a proof of Theorem~\ref{th1}.

\begin{proof}[Proof of Theorem~\ref{th1}]
For $j\in \Z_q$, we denote	
$\Lambda_{j}(x)=\sum\limits_{\ell=\frac{q+1}{2}}^{q-1}D_{\ell+j}(x),$ where all the subscripts are  understood modulo $q$ here. 
Note that  $\Lambda_{0}(x)=\sum\limits_{\ell=\frac{q+1}{2}}^{q-1}D_{\ell}(x)$ is the generating polynomial of the sequence $\mathbf{s}$ exactly. 
Now we claim that $\Lambda_{j}(\gamma)\neq 0$ for all $j\in \Z_q$. 	

We first prove $2\not\in{D_{0}}$ under the condition that 
$2^{q-1}\not \equiv 1 \pmod{q^2}$. 
Suppose that $2\in{D_{0}},$ i.e., $\psi(2)=\frac{2^{\varphi(pq)}-1}{pq}=0 \pmod{pq}$ according to the definition of Euler quotients.  This implies that
$$\frac{2^{\varphi(pq)}-1}{pq}\equiv 0 \pmod{q}$$
and thus
$$2^{\varphi(pq)}= (2^{q-1})^{p-1} \equiv  1 \pmod{q^2}.$$
This means that the order of  $2^{q-1} \pmod{ q^2}$ is a factor of  $p-1.$ However, it follows from  $2^{q-1}\not \equiv 1 \pmod{q^2}$ that the order of  $2^{q-1} \pmod{ q^2}$ is exactly equal to $q$. 
This implies that $q$ divides $p-1,$ which contradicts the condition that $p< q.$
Hence, there exists some fixed nonzero  $\sigma\in \Z_q$ such that 
$2\in{D_{\sigma}}$.

In the following we argue by contradiction. 
Assume that there exists some $j_0\in \Z_q$ such that
$\Lambda_{j_0}(\gamma)=0,$ where $\gamma$ is a $pq^{2}$-th primitive root of unity.  
By Lemma~\ref{lma4} we get $$0=\Lambda_{j_0}(\gamma)^{2^{i}}=\Lambda_{j_0}(\gamma^{2^{i}})=\Lambda_{j_0+i\sigma}(\gamma)$$
for any $i\in \Z_q$. According to  $\sigma \neq 0$ in $\Z_q$,  the number $j_0+i\sigma$ runs through $\Z_q$ when $i$ runs through $\Z_q$. 
This means that $\Lambda_{j}(\gamma)=0$ for all $ j\in \Z_q.$ In particular, we can choose  $\Lambda_{0}(\gamma)=0$. 

For any $v\in{D_{j}}$ with $ j \in \Z_q,$  it follows from Lemma~\ref{lma4} that 
$$\Lambda_{0}(\gamma^{v})=\sum\limits_{\ell=
	\frac{q+1}{2}}^{q-1}{D_{\ell}}(\gamma^{v})=\sum\limits_{\ell=
	\frac{q+1}{2}}^{q-1}{D_{\ell+j}}(\gamma)=
\Lambda_{j}(\gamma)=0.$$

Note that $\Z_{pq^2}^*=\bigcup\limits_{j=0}^{q-1}D_{j}.$
It is immediate that $\Lambda_{0}(\gamma^{v})=0$ for any $v\in{\Z_{pq^2}^*}.$ Thus the cyclotomic polynomial $\Phi_{pq^2}(x)$ divides $\Lambda_{0}(x)$.
By Lemma~\ref{lma8}, we see that  $\Phi_{q^2}(x)$ divides $\Lambda_{0}(x).$
Then $\Phi_{pq^2}(x)\Phi_{q^2}(x)$ divides $\Lambda_{0}(x)$ since $\gcd(\Phi_{pq^2}(x),\Phi_{q^2}(x))=1.$
On the basis of Exercise 2.57 of~\cite{Lidl1986}, we have 
$$\Phi_{pq^2}(x)\Phi_{q^2}(x)=\Phi_{q^2}(x^{p})=\Phi_{q}(x^{pq})=\sum\limits_{j=0}^{q-1}x^{jpq}.$$
We write
$$\Lambda_{0}(x)\equiv \Phi_{q}(x^{pq})\pi(x) \pmod{x^{pq^2}-1}.$$
Note that 
\begin{equation*}
	\begin{aligned}
x^{pq}\Phi_{q}(x^{pq})=& x^{pq}\sum\limits_{j=0}^{q-1}x^{jpq}\equiv\sum\limits_{j=0}^{q-1}x^{jpq}\\
\equiv&\Phi_{q}(x^{pq}) \pmod{x^{pq^2}-1}.
\end{aligned}
\end{equation*}

We can restrict $\deg{\pi(x)}< pq$  and thus $\pi(x)$ can be written as $\pi(x)=\sum\limits_{i=0}^{t-1}x^{\nu_i},$ where $0\leq \nu_0< \nu_1< \cdots < \nu_{t-1}<pq.$
Then
\begin{equation*}
	\begin{aligned}
	\Lambda_{0}(x)\equiv& \pi(x)\Phi_{q}(x^{pq}) 
	\equiv\sum\limits_{i=0}^{t-1}x^{\nu_i}\sum\limits_{j=0}^{q-1}x^{jpq}\\
	\equiv& \sum\limits_{i=0}^{t-1}\sum\limits_{j=0}^{q-1}x^{\nu_{i}+jpq} \pmod{x^{pq^2}-1}.
	\end{aligned}
\end{equation*}
However $\Lambda_{0}(x)$ has $\frac{1}{2}(p-1)(q-1)^2$ terms and $\sum\limits_{i=0}^{t-1}\sum\limits_{j=0}^{q-1}x^{\nu_{i}+jpq}$ has $qt$ terms, which is a contradiction since the prime $q$ does not divide $\frac{1}{2}(p-1)(q-1)^2$.
It follows that $\Lambda_{j}(\gamma)\neq 0$ for any $j\in \Z_q.$

This implies that $\Lambda_{j}(\gamma^{v})\neq 0$ for all $v\in{\Z_{pq^2}^*}$ and $0\leq j < q.$ In particular, we have $\Lambda_{0}(\gamma^{v})\neq 0$ for all $v\in{\Z_{pq^2}^*}$. 
By Lemma~\ref{lma8}, for $v \in \Z_{pq^2}$ and $0\leq j < q$ we get
\[\Lambda_{0}(\gamma^{v})= \left\{\begin{array}
{l@{\quad\quad}l}
0, &\textrm{if} ~~~ \gcd{(v, pq^2)}\in{\{p, pq, q^2, pq^2\}}, \\
\frac{q-1}{2},  &\textrm{if} ~~~ \gcd{(v, pq^2)}=q.
\end{array}\right.\]
This implies that $(x-1) {(\Phi_{p}(x)\Phi_q(x)\Phi_{q^2}(x))}$ divides $\Lambda_0(x)$ if $q\equiv 3 \pmod 4$. Hence the minimal polynomial in the case that $q\equiv 3 \pmod 4$ is 
$$M_{\sq}(x)=\frac{x^{pq^2}-1}{(x-1) {\Phi_{p}(x)\Phi_q(x)\Phi_{q^2}(x)}}=\Phi_{pq^2}(x)\Phi_{pq}(x)$$
by using the basic properties of cyclotomic polynomials.  In a similar manner, if $q\equiv 1 \pmod{4}$, then $\Phi_{pq}(x)$ divides $\Lambda_0(x)$. This yields that the minimal polynomial in the case that $q\equiv 1 \pmod 4$ is  
$$M_{\sq}(x)=\frac{x^{pq^2}-1}{(x-1) {(\Phi_{p}(x)\Phi_q(x)\Phi_{pq}(x))\Phi_{q^2}(x))}}=\Phi_{pq^2}(x).$$
Note that the linear complexity of $\sq$ is equal to the degree of the minimal polynomial of the sequence and so the third assertion in Theorem~\ref{th1} follows.   
This completes the whole proof of Theorem~\ref{th1}.$ $
\end{proof}

In the following, we will give a small example for confirming our main results. 
\begin{example}
	Let $p=3$ and $q=7$. The least period of the binary threshold sequence  $\sq$ derived from modulo $pq$ is $147$. 	 The sequence $\sq$ in one period is 
	\begin{equation*}
	\begin{split}
	\{0&, 0, 0, 0, 1, 1, 0, 0, 1, 0, 0, 0, 0, 0, 0, 0, 0, 0, 0, 0, 0, 0, 1, 1, 0, 1,\\ 
	  1&, 0, 0, 0, 0, 0, 0, 0, 1, 0, 0, 0, 0, 0,1,1, 0, 0, 1, 0, 1, 0, 0, 0, 0, 0, \\
	  1&,1, 0,0, 0, 0, 1, 0, 0, 0, 0, 0, 1, 0, 0, 0, 	0, 0, 0, 1, 0, 1, 1, 0, 1, 0, \\
	0&, 0, 0, 0,0,1, 0, 0, 0, 0, 0, 1, 0, 0, 0, 0, 1, 1, 0, 0, 0, 0, 0, 1, 0, 1, \\
	0&, 0, 1, 1, 0, 
	0, 0,0, 0, 1, 0, 0, 0, 0, 0, 0, 0, 1, 1, 0, 1, 1, 0, 0, 0, 0,\\
	 0&, 0, 0, 0, 0, 0, 0,0, 0,1, 0, 0, 1, 1, 0, 0, 0 \}.	\end{split}
	\end{equation*}	 
	The minimal polynomial of the sequence $\sq$ over $\F_2$ is	
	\begin{equation*}
	\begin{split}
	x^{96}& + x^{95} + x^{93} + x^{92} + x^{90} + x^{89} + x^{87} + x^{86} + x^{84} + x^{83} \\
	 +& x^{81} + x^{80} +x^{78} + x^{77}+ x^{75}+ x^{74} + x^{72} + x^{71} + x^{69}  \\
	 +&x^{68} + x^{66} + x^{65} + x^{63} + 	x^{62} + x^{60} + x^{59} + x^{57} + x^{56} \\
	+&x^{54}+ x^{53} + x^{51} + x^{50} + x^{48} + x^{46} + x^{45} + x^{43} + x^{42} \\
	+& x^{40} + x^{39} + x^{37} + x^{36} + x^{34} + x^{33}+x^{31} + x^{30} + x^{28}\\
	+&  x^{27} + x^{25} + x^{24} + x^{22} + x^{21} + x^{19} + x^{18}
	+x^{16} + x^{15}\\
	+& x^{13} + 	x^{12} + x^{10} + x^{9} + x^{7} + x^6 + x^{4} + x^{3} + x + 1
	\end{split}
	\end{equation*}	 
	and  so the linear complexity of this sequence is exactly $(p-1)\cdot(q^2-1)=2\cdot 48=96.$
	
\end{example}

\section{\textsf{Conclusion Remarks}}\label{sec4}
In this paper, we determined the  linear complexities of a class of binary sequences with period $pq^2$ based on the  Euler quotients modulo $pq$. In addition, the proposed sequences have a good balance asymptotically if the prime $p$ tends to infinity, i.e., the number of 1's is asymptotically equal to the number of 0's in one period if $p$ tends to infinity. {\color{black}
Finally, there are several unsolved problems about the proposed sequence. Below are some of them. 
\begin{itemize}
	\item  Determine $k$-error linear complexity of the proposed sequence. This problem is closely related to the stability of linear complexity of the proposed sequence from a cryptographic point of view.
	\item  Explore whether this family of sequences derived from the Euler quotient modulo $pq$ can induce more optimal families of perfect polyphase sequences similar to~\cite{Song2016}. 
	\item  Regard the proposed sequence as the sequence over the ring $\Z_4$ or other finite fields with odd prime characteristics and analyze its cryptographic properties ( for example, linear complexity and $k$-error linear complexity). 
	\item Determine the autocorrelation of the proposed sequence although it may be hard to tackle this problem. 
	\item Analyze properties of sequences derived from the Euler quotient with other modulus (for example, the modulus is equal to the product of more odd distinct primes). 
\end{itemize}
}
%
%

\section*{Acknowledgment}
The authors would like to thank the associate editor Prof. Sihem Mesnager 
and the two anonymous referees for
their helpful suggestions which improved this manuscript.
	
	Part of the work were done while  the first and third authors were Visiting  Scholars at Clemson University.

%
\IEEEpeerreviewmaketitle


\ifCLASSOPTIONcaptionsoff
  \newpage
\fi


\vfill

\begin{IEEEbiographynophoto}{Jingwei Zhang}
received the B.S. degree (2002) from Department
	of Mathematics, Hunan University of Sicence and Technology, 
	the M.S. degree (2005)  in Department of Mathematics and the Ph.D. degree (2010) in Department of Electrical Engineering from Sun Yat-sen University,
	Guangzhou, P.R. China. She currently works in school of Informatics
	in Guangdong University of Finance $\&$ Economics,Guangzhou, P.R. China.  Her research interest
	includes algebraic coding theory, algebraic decoding algorithms and sequences.
\end{IEEEbiographynophoto}
\begin{IEEEbiographynophoto}{Chang-An Zhao}
	received his BS in Department of Electrical Engineering in 2001,  MS in Department of Mathematics in 2005,  and PhD  in Department of  Computer Science in 2008 respectively from Sun Yat-sen University, Guangzhou, P.R. China. From 2008 to 2013, he worked in Department of computer Science in Guangzhou University, P.R. China. In 2013, he joined Department of Mathematics  as a lecturer in Sun Yat-sen University, P.R. China and was promoted to be an associate professor in 2016.  His research interest
	lies in sequences, elliptic curve cryptography and coding theory.
\end{IEEEbiographynophoto}
\begin{IEEEbiographynophoto}{Shuhong Gao}
Shuhong Gao received his BS (1983) and MS (1986) from Department of Mathematics, Sichuan University, China, and PhD (1993) from Department of Combinatorics and Optimization, University of Waterloo, Canada.  From 1993 to 1995, he was an NSERC Postdoctoral Fellow in Department of Computer Science, University of Toronto, Canada. He joined Clemson University in USA in 1995 as an assistant professor in Mathematical Sciences, and was promoted to associate professor in 2000 (with early tenure) and to full professor in 2002.  Professor Gao's research interests include coding theory, cryptography, blockchains, machine learning, quantum computing, and computational algebraic geometry. More information about his research and  teaching can be found at Applicable Algebra Lab: https://www.ces.clemson.edu/aca/.
\end{IEEEbiographynophoto}

\end{document}